\documentclass[a4paper,12pt]{article}
\usepackage{latexsym,amssymb,amsmath,amsthm}
\usepackage[dvips]{graphicx}
\newtheorem{theorem}{Theorem}
\newtheorem{lemma}{Lemma}
\newtheorem{corollary}[theorem]{Corollary}
\newtheorem{proposition}{Proposition}
\newtheorem{remark}{Remark}
\newtheorem{definition}{Definition}

\newcommand{\bs}[1]{\boldsymbol{#1}}
\newcommand{\im}{\bs{\rm i}}

\title{{\Large {\bf Quantum graph walks I: \\ mapping to quantum walks}}}

\author{ 
{\small 
Yusuke Higuchi,$^{1}$ %authorname1
\footnote{higuchi@cas.showa-u.ac.jp % email adress1
}\quad 
Norio Konno,$^{2}$ %authorname1
\footnote{konno@ynu.ac.jp % email adress1
}\quad 
Iwao Sato,$^{3}$ %authorname2
\footnote{isato@oyama-ct.ac.jp % email adress2
}\quad  
Etsuo Segawa,$^{4}$ %authorname3
\footnote{Corresponding author: e-segawa@m.tohoku.ac.jp % email adress3
} 
}\\ 
{\scriptsize $^{1}$ 
Mathematics Laboratories, College of Arts and Sciences, Showa University%shozoku
}\\
{\scriptsize 
Fuji-Yoshida, Yamanashi 403-005, Japan
}\\
{\scriptsize $^{2}$ 
Department of Applied Mathematics, Faculty of Engineering, Yokohama National University%shozoku
}\\
{\scriptsize Hodogaya, Yokohama 240-8501, Japan%juusho
} \\
{\scriptsize $^3$ 
Oyama National College of Technology, %shozoku
}\\
{\scriptsize Oyama, Tochigi 323-0806, Japan%juusho
} \\
{\scriptsize $^4$ 
Graduate School of Information Sciences, Tohoku University%shozoku
}\\
{\scriptsize Aoba, Sendai 980-8579, Japan%juusho
}\\ 
} 

\date{\empty }
\begin{document}
\maketitle

\par\noindent
\begin{small}
\par\noindent
{\bf Abstract}. 
%Both quantum walk and quantum graph describe some quantum dynamics on graphs. 
%They have been studied independently. 
We clarify that coined quantum walk is determined by only the choice of local quantum coins. 
To do so, we characterize coined quantum walks on graph by disjoint Euler circles with respect to symmetric arcs. 
In this paper, we introduce a new class of coined quantum walk by a special choice of quantum coins determined by 
corresponding quantum graph, called quantum graph walk. 
We show that a stationary state of quantum graph walk describes the eigenfunction of the quantum graph. 
%abst
\footnote[0]{
{\it Key words and phrases.} Quantum walk, quantum graph 
 %key wards 
}

\end{small}

\setcounter{equation}{0}
%%%%%%%%%%%%%%%%%%%%%%%%%%%%%%%%%%%%%%%%%%%%%%%%%%%%%%%%%%%%%%%%%%%%%%%%%%%%%%%%%%%%%%%%%%%%%%%%%%%%%%%%%%%
%%%%%%%%%%%%%%%%%%%%%%%%%%%%%%%%%%%%%%%%%%%%%%%%%%%%%%%%%%%%%%%%%%%%%%%%%%%%%%%%%%%%%%%%%%%%%%%%%%%%%%%%%%%
\section{Introduction}
%%%%%%%%%%%%%%%%%%%%%%%%%%%%%%%%%%%%%%%%%%%%%%%%%%%%%%%%%%%%%%%%%%%%%%%%%%%%%%%%%%%%%%%%%%%%%%%%%%%%%%%%%%%
%%%%%%%%%%%%%%%%%%%%%%%%%%%%%%%%%%%%%%%%%%%%%%%%%%%%%%%%%%%%%%%%%%%%%%%%%%%%%%%%%%%%%%%%%%%%%%%%%%%%%%%%%%%
The quantum walk has been intensively studied from various kinds of view points, 
since it was treated as a part of quantum algorithm in quantum information \cite{ABNVW} and its strong efficiency 
to so called quantum speed up search was shown (see \cite{Ambainis_Rev} and its references).  
For example, the Anderson localization~\cite{OK,Hanover,Joye}, stochastic behaviors comparing with random walks~\cite{KonnoBook}, 
spectral analysis of the unit circle~\cite{CGMV} in relation to the CMV matrix~\cite{CMV}, graph isomorphic problem~\cite{EHSW}, experimental implementation~\cite{KFCSA}, and so on. 
Stanly Gudder is one of the originators of discrete-time quantum walk on graph~\cite{Gudder} (1988). 
At first, for simplicity, 
let us consider the walk on one dimensional lattice  following the Gudder's book. 
In this walk, each vertex has the left and right chiralities. 
The total state space here is spanned by the canonical basis corresponding to these chiralities , 
that is, $\{|j,R\rangle, |j,L\rangle : j\in \mathbb{Z}\}$. 
Let $\psi^{(L)}_n(j)$ and $\psi^{(R)}_n(j)$ as scaler valued left and right amplitudes at time $n$ position $x\in \mathbb{Z}$, respectively. 
The time evolution is given by the recurrence relations as follows : 
\begin{align}
\psi_{n}^{(R)}(j) &= a \psi_{n-1}^{(R)}(j-1)+\im b \psi_{n-1}^{(L)}(j+1), \label{tm1}\\
\psi_{n}^{(L)}(j) &= \im b \psi_{n-1}^{(R)}(j-1)+ a \psi_{n-1}^{(L)}(j+1), \label{tm2}
\end{align}
where $a,b\in \mathbb{R}$ with $a^2+b^2=1$. 
An equivalent expression for this time evolution, which will be important to our paper, is that: 
putting $\bs{\psi}_n(j)={}^T [\psi_{n}^{(R)}(j),\psi_{n}^{(L)}(j)]$, then 
\begin{equation}
\bs{\psi}_n(j)=Q\bs{\psi}_{n-1}(j-1)+P\bs{\psi}_{n-1}(j+1), 
\end{equation}
where 
\[ P=\begin{bmatrix} 0 & \im b \\ 0 & a \end{bmatrix},\;\;Q=\begin{bmatrix} a & 0 \\ \im b & 0 \end{bmatrix}. \]
We can interpret the quantum walk as a walk which has matrix valued weights $P$ and $Q$ associated with moving to left and right, respectively. 
Anyway, equations (\ref{tm1}) and (\ref{tm2}) imply that
\begin{equation}
\psi_{n+1}^{(J)}(j)+\psi_{n-1}^{(J)}(j) = a \left\{ \psi_{n}^{(J)}(j-1)+\psi_{n}^{(J)}(j+1) \right\},\;\;(J\in \{L,R\})
\end{equation}
which is a discrete-analogue of the mass less Klein-Gordon equation: 
\[ \frac{\partial^2 \psi_t(x)}{\partial t^2}=a \frac{\partial^2 \psi_t(x)}{\partial x^2}. \]
This is considered as one of the motivations for introducing this walk. 

We show another reason for why the total space of QW is described by not $\mathbb{Z}$ but $\mathbb{Z}\times \mathbb{C}^2$. 
An idea which is across our mind immediately to accomplish a quantization of a random walk on one dimensional lattice 
may be as follows: 
the probabilities $p$ and $1-p$ with $p>0$ that 
moving to left and right in random walk at each time step are replaced with some complexed valued weights $\alpha$ and $\beta$ 
so that its one step time operator is unitary. 
However we can easily see that the postulate of its unitarity implies $\alpha\overline{\beta}=0$. 
Thus the walk becomes quite trivial one, that is it always goes to the same direction. 
It is the no-go lemma~\cite{Meyer} of quantum walk. So we need left and right chiralities at each vertex in one dimensional lattice. 
Reference \cite{Severini} gives more detailed discussion for a general graphs around here. 

Now in the next, we consider the walk extending to a general graph. 
Let $\mathcal{G}(V,E)$ be a graph with vertex set $V(\mathcal{G})$ and edge set $E(\mathcal{G})$. 
In this paper, we denote the edge $e\in E(\mathcal{G})$ between vertices $u$ and $v$, as $e=\{u,v\}=\{v,u\}$. 
For $u\in V$, we define $N(u)=\{v\in V: \{u,v\}\in E \}$, and $d_u$ is degree of $u$, that is, $d_u=|N(u)|$. 
We define the set of symmetric arcs $D(G)$ as 
$\{(u,v)\in V(\mathcal{G})\times V(\mathcal{G}) : \{u,v\}\in E(\mathcal{G}) \}$. 
We denote arc $a=(u,v)\in D(G)$ as 
$o(a)=u$ and $t(a)=v$, where $o(a)$, and $t(a)$ are the origin and the terminus of $a$, respectively. 
For $a=(u,v)\in D(G)$, we denote $a^{-1}$ as $(v,u)$. 
The quantum walk on $\mathcal{G}(V,E)$ introduced by Gudder (1988) is defined as an analogue of the one dimensional lattice case. 
%%%%%%%%%%%%%%%%%%
\begin{definition}\label{od} (Definition of quantum walk)
\begin{enumerate}
\item Total space: Let $\mathcal{H}$ be the total space of quantum walk. 
\[ \mathcal{H}=\ell^2(D(\mathcal{G}))=\mathrm{span}\{|u,v\rangle: (u,v)\in D(\mathcal{G})\}.  \]
Let $\mathcal{H}=\oplus \sum_{u\in V(\mathcal{G})}\mathcal{H}_u$ with 
$\mathcal{H}_u\cong \mathrm{span}\{ |u,v\rangle;v\in N(u) \}$. We denote the canonical basis of the subspace $\mathcal{H}_u$ as 
$\{|\bs{e}^{(u)}_v\rangle; v\in N(u)\}$. 
\item Time evolution: To every $(u,v)\in D$, 
we assign a non-trivial linear map $\mathcal{H}_u\to \mathcal{H}_v$ with its matrix representation $W_{(u,v)}$ so that 
$|D|\times |D|$ matrix on $\mathcal{H}$, $U$, defined by 
\[ \langle s,t|U|u,v \rangle=\mathbf{1}_{\{(u,s)\in D\}} \langle \bs{e}^{(s)}_t|W_{(u,s)}|\bs{e}^{(u)}_v\rangle \] 
is a $|D|$-dimensional unitary matrix. 
The time evolution is the iteration of the unitary operator $U$ with an initial state $\Psi_0\in \mathcal{H}$ with $||\Psi_0||=1$ 
such that 
	$\Psi_0 \stackrel{U}{\mapsto} \Psi_1 \stackrel{U}{\mapsto} \Psi_2\stackrel{U}{\mapsto}\cdots$, where $\Psi_j=U^j\Psi_0$. 
\item Measure\footnote{In this paper, we slightly modify the original definition of measure in \cite{Gudder} to emphasize a correspondence to the random walk on the same graph. 
In the original definition, indeed, $\Omega_n=\{(q_0,\dots,q_n)\in D(\mathcal{G})^n: t(q_j)=o(q_{j+1})\}$. 
%For $A\in 2^{\Omega_n}$ $\mu_n(A)$ is defined by $\mu_n(A)=|\sum_{\xi=(q_0,\dots,q_n)\in A} \langle q_{n}|W_{o(q_{n-1})t(q_{n-1})}}\cdots W_{o(q_{0})t(q_{0})}}|q_0^{-1}\rangle|^2$.
}
: Denote $\Omega_n$ as the set of all the $n$-truncated possible paths from a vertex $o\in V(\mathcal{G})$. 
The measure $\mu_n: 2^{\Omega_n} \to [0,1]$ is defined as follows: 
for $A\in 2^\Omega_n$, 
\[ \mu_n^{\varphi}(A)=\left|\left|\sum_{\xi=(o,\xi_2\dots,\xi_{n})\in A} W_{(\xi_{n-1},\xi_n)}\cdots W_{(\xi_2,\xi_3)}W_{(o,\xi_2)}\bs{\varphi}\right|\right|^2, \]
where $\bs{\varphi}$ is a vector in $\mathcal{H}_{o}$. 
\end{enumerate}
\end{definition}
%%%%%%%%%%%%%%%%%
\begin{remark}
We can see this is an extension to a general graph of the one dimensional case in the following sense: 
for each arc $(i,j)$ with $|i-j|=1$, under the following one-to-one correspondence between the canonical basis, 
$|j,j-1\rangle \leftrightarrow |j,L\rangle$, $|j,j+1\rangle \leftrightarrow |j,R\rangle$, 
the weights of moving left and right at each vertex are $W_{(j,j+1)}=Q$ and $W_{(j,j-1)}=P$, ($j\in \mathbb{Z}$). 
\end{remark}
%%%%%%%%%%%%%%%%%%
For $u\in V(\mathcal{G})$, the measure of $A_u=\{\xi\in \Omega_n: \xi_n=u\}\in 2^{\Omega_n}$ gives a distribution since 
$\sum_{u\in V(\mathcal{G})}\mu_n(A_u)=1$, and $\mu_n(A_u)\in [0,1]$. 
We define the finding probability of quantum walk at time $n$, position $u$ by $\mu_n(A_u)$. 
In this paper, we classify a special case of the discrete-time quantum walks in Def.\ref{od}, so called coined quantum walk 
which is defined by introducing local unitary operator (called {\it quantum coin}) for each $u\in V(\mathcal{G})$ on $\mathcal{H}_u$. 
In \cite{Watrous}, we can see the original form of the Grover walk on general graphs which are most intensively studied by many researchers. 
The Grover walk is in a special class of coined quantum walks called ``A-type quantum walks with flip flop shift" in this paper. See Sect.~\ref{sec:2} for its detailed definition.
We clarify that the investigation of A-type quantum walk is essential to study of coined quantum walk. 
More concretely, we find that for fixed local quantum coins, we can express any coined quantum walks by an A-type quantum walk with flip flop shift with a permutation (Theorem \ref{thm1}). 
Thus a choice of local quantum coins determines the coined quantum walk. 

By the way, a quantum graph is a system of 
a linear Schr{\"o}dinger equations on each Euclidean edge with boundary conditions at each joined part, i.e., vertex. 
The quantum graph is determined by triple of sequences of parameters $(\bs{L},\bs{\lambda},\bs{A})$ with respect to Euclidean edge lengths, boundary conditions, and vector potentials on edge, respectively. 
See Sect.~4.1 for the detailed setting of the quantum graph. 
Quantum graphs have been studied from varions fields of view. 
For the review and books on quantum graphs, see \cite{ES,Kuch,Smilansky}, for examples. 

In this paper, we apply the formulation of quantum graphs according to Smilansky and his group \cite{Smilansky,Smilansky2}. 
Anyway, what is the solution (eigenfunction) for the system of Schr{\"o}dinger equations which satisfy the boundary conditions simultaneously ? 
To answer it, in this paper, we define a coined quantum walk, $U^{(\bs{L,\lambda,A})}$, by a special choice of local quantum coins determined by corresponding quantum graph. 
We call this walk {\it quantum graph walk} whose more detailed definition is denoted in Sect.~4.2. 
The following result is our main theorem: 
%%%%%%%%%%%%%%
\begin{theorem}\label{mainthm}
The quantum graph walk with parameters $(\bs{L},\bs{\lambda},\bs{A})$ has non-trivial eigenfunction satisfying all the boundary conditions at vertices simultaneously if and only if 
\[ U^{(\bs{\lambda,L,A})}\bs{a}_*(k)=\bs{a}_*(k). \]
Here a linear transformation of $\bs{a}_*(k)$ is the eigenfunction of the quantum graph. (See Eq.~(\ref{wq}) for an explicit expression for the linear transformation.)
\end{theorem}
%%%%%%%%%%%%%
This paper is organized as follows. 
Section 2 is devoted to special quantum walks called coin-shift type quantum walks. 
The time evolution of coin-shift type quantum walk $U$ has two stages; coin operator $C$, and the shift operator $S$. 
In the coin-shift type quantum walk, the walk is characterized by the choice of coin operator. 
The next of two sections (Sects. 3 and 4), we treat two special classes of the discrete-time quantum walk. 
The first is the Szegedy walk introduced by Szegedy\cite{Szegedy} (2004), which is induced by a transition matrix of a random walk on the same graph. 
One of the strong facts is that a main part of eigensystems of the Szegedy walk is obtained once we know the eigensystem of the corresponding random walk. 
The Szegedy walk induced by the symmetric random walk, that is, a walker moves to a neighbor uniformly, becomes the famous Grover walk which is most intensively studied in the view point of quantum information. 
We have already know the eigensystem of the Szegedy walk is described by the spectrum of corresponding random walk. 
The second one is the {\it quantum graph walk} induced by a quantum graph \cite{Smilansky,Smilansky2}. 
As we have seen in Theorem \ref{mainthm}, we find that the Schr{\"o}dinger equation has non trivial solution iff the quantum graph walk has stationary amplitude. 
Moreover in the Neumann boundary condition, in the limit of edge length zero, we can see the Grover walk again. 
We give its proof and an expression for the eigenequation of $U^{(\bs{\lambda,L,A})}$ which is reduced to vertex size $|V|$ from square of edge size $2|E|$. 
The common part of the Szegedy walk and the quantum graph walk is the Grover walk. 
As far as we know, Ref.~\cite{Tanner} is the first paper which suggests a relation between the quantum graph and quantum walk. 
We more clarify and refine its relationship in this paper. 
One of the most important suggestions for a usefulness of mapping to quantum walks is Ref.~\cite{SandS}: Schanz and Smilansky~\cite{SandS} (2000) 
have already shown a localization 
of the quantum graph on random environment of $\mathbb{Z}$ mapping to a quantum scattering evolution 
which can be interpreted as nothing but now a day a spatial disorded discrete-time quantum walk with some modifications. 
Localization is a recent hot topic of quantum walks. For example, Refs.~\cite{OK,Hanover,Joye,KonnoBook,CGMV,KonnoLoc,KLS}. 
They gave a strictly positive return probability for annealed law by a combinatorial analysis before the quantum walks were so intensively studied. 

%%%%%%%%%%%%%%%%%%%%%%%%%%%%%%%%%%%%%%%%%%%%%%%%%%%%%%%%%%%%%%%%%%%%%%%%%%%%%%%%%%%%%%%%%%%%%%%%%%%%%%%%%%%
%%%%%%%%%%%%%%%%%%%%%%%%%%%%%%%%%%%%%%%%%%%%%%%%%%%%%%%%%%%%%%%%%%%%%%%%%%%%%%%%%%%%%%%%%%%%%%%%%%%%%%%%%%%
\section{Quantum walks on graph: reconsideration}\label{sec:2}
%%%%%%%%%%%%%%%%%%%%%%%%%%%%%%%%%%%%%%%%%%%%%%%%%%%%%%%%%%%%%%%%%%%%%%%%%%%%%%%%%%%%%%%%%%%%%%%%%%%%%%%%%%%
%%%%%%%%%%%%%%%%%%%%%%%%%%%%%%%%%%%%%%%%%%%%%%%%%%%%%%%%%%%%%%%%%%%%%%%%%%%%%%%%%%%%%%%%%%%%%%%%%%%%%%%%%%%
In this paper, we treat a connected and simple graph, that is, without self loops and multiedges. 
A path is a sequence of vertices of $\mathcal{G}$, $u_1,u_2,\dots,u_n$ with $(u_i,u_{i+1})\in D(\mathcal{G})$.  
The line digraph of $\overrightarrow{L}\mathcal{G}(V,A)$ with the vertex set $V$ and arc set $A$ is defined as follows: 
\[ V\left(\overrightarrow{L}\mathcal{G}\right)=D(\mathcal{G}),
\;\;A\left(\overrightarrow{L}\mathcal{G}\right)=\left\{\big((u,v),(v',w)\big)\in V\left(\overrightarrow{L}\mathcal{G}\right)^2: v=v'\right\}. \]
A cycle in a graph $\mathcal{G}$ is a path $u_1,u_2,\dots,u_{n},u_1$ with $(u_j,u_{j\oplus_n 1})\in D(H)$, where $l \oplus_n m=\mathrm{mod}((l+m),n)$. 
In particular, if all the $u_j$'s in the sequence are distinct, then we call it essential cycle. 
Note that if a cycle $(u_1,u_2),(u_2,u_3),\dots,(u_n,u_1)$ with $u_j\in V(\mathcal{G})$ in the line digraph $\overrightarrow{L}\mathcal{G}$ 
is essential, then the sequence $u_1,u_2,\dots,u_n,u_1$ of the original graph $\mathcal{G}$ is also cycle, however 
its essentiality is not ensured.  
On the other hand, if a sequence $u_1,u_2,\dots,u_n,u_1$ is essential in $\mathcal{G}$, then 
$(u_1,u_2),(u_2,u_3),\dots,(u_n,u_1)$ is also essential in $\overrightarrow{L}\mathcal{G}$. 
\begin{definition}
Let $\pi$ be a partition on $\overrightarrow{L}\mathcal{G}$ such that 
\begin{equation}
\pi: \overrightarrow{L}\mathcal{G}\to \{C_1,C_2,\dots,C_r\},
\end{equation}
where $C_j$ is an essential cycle of $\overrightarrow{L}\mathcal{G}$ and 
$\bigcup_{j=1}^r V(C_j)=V(\overrightarrow{L}\mathcal{G})$, $V(C_i)\cap V(C_j)=\emptyset$
for $i\neq j$. We denote the set of all the such partitions as $\Pi_{\mathcal{G}}$. 
\end{definition}
\begin{remark}
The following partition called ``flip flop partition"  belongs to $\Pi_{\mathcal{G}}$ for every undirected graph. 
\begin{align}
\pi_{ff}: \overrightarrow{L}\mathcal{G} &= \{ C_1,\dots,C_{|E(\mathcal{G})|} \}, 
\end{align}
where $V(C_j) = \{ e_j,e^{-1}_j \}$, and $A(C_j)=\{ (e_j,e^{-1}_j), (e^{-1}_j,e_j) \}$ for $e_j\in \mathcal{D}(\mathcal{G})$.
\end{remark}
The partition $\pi$ gives a way to decompose the graph $\mathcal{G}$ into mutually disjoint Euler circles with respect to arcs. 
Let $\Pi_u$ be the set of all one-to-one correspondence between 
\[ \{ |\bs{e}^{(u)}_v\rangle;v\in N(u) \} \leftrightarrow \{ |\bs{e}^{(v)}_u\rangle;v\in N(u) \}. \]
The former one corresponds to out-neighbor of $u$, and later one is in-neighbor of $u$. 
There are many partitions in $\Pi_\mathcal{G}$ in fact $|\Pi_\mathcal{G}|=\prod_{u\in V} d_u !$ since 
\[ \bigotimes_{u\in V(\mathcal{G})}\Pi_u \cong \Pi_\mathcal{G}.  \]
Since the out- and in-degrees of all the vertices in $C_j$ are $1$, 
we can define the following map $f_\pi$ (see also Fig.~1.): 
\begin{definition}
For $\pi\in \Pi_{\mathcal{G}}$ with $\overrightarrow{L}\mathcal{G}\stackrel{\pi}{\mapsto} \{C_1,\dots, C_r\}$, we define 
\begin{equation}
f_\pi: V\left(\overrightarrow{L}\mathcal{G}\right)\to V(\mathcal{G})
\end{equation}
such that for any $(i,j)\in V(\overrightarrow{L}\mathcal{G})$, 
\begin{equation}
 \big( (i,j),(j,f_\pi(i,j)) \big)\in \bigcup_{j=1}^r A(C_j) .
\end{equation}
\end{definition}
\begin{figure}
\begin{center}
\includegraphics[width=10cm]{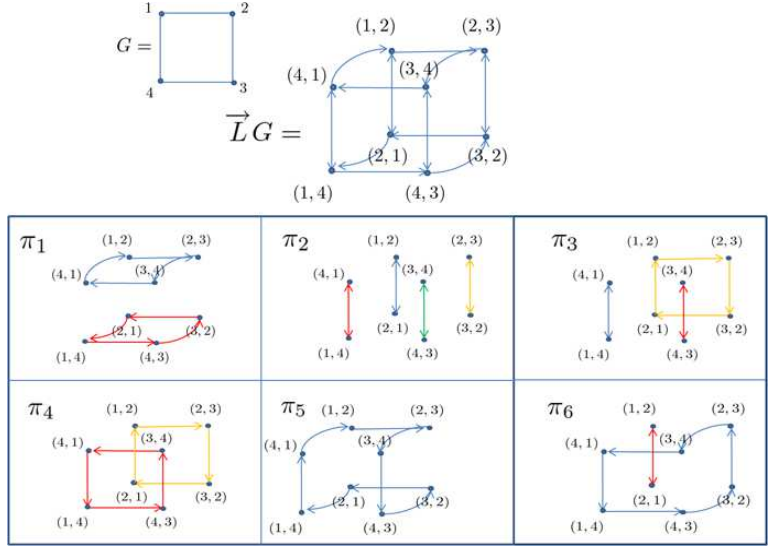}
\caption{{\scriptsize Decomposition into mutually disjoint Euler circles: 
The $2^4=16$ partitions of the circle with four vertices are classified into above $6$ patterns $[\pi_j]$ ($j=1,2,\dots,6$) 
with respect to automorphism. 
The cardinalities of each conjugacy classes $|[\pi_j]|$ $(j=1,2,\dots,6)$ are $1$, $1$, $4$, $2$, $4$, $4$, respectively. 
Indeed, $\Pi_G=\sum_{j}|[\pi_j]|=16$. 
We see $f_{\pi_1}(1,2)=3$, $f_{\pi_1}(3,4)=1$, $f_{\pi_2}(1,2)=1$, $f_{\pi_2}(3,4)=3$, $f_{\pi_3}(1,2)=3$, $f_{\pi_3}(3,4)=3$ and so on. 
 }  }
\end{center}
\end{figure}
From now on, we explain a special class of quantum walk called {\it coined quantum walks} on graph $\mathcal{G}$ under these setting. 
We choose a partition $\pi$ from $\Pi_{\mathcal{G}}$, 
and a sequence of unitary operators $\{H_j\}_{j=1}^{|V|}$, 
where $H_j$ is a $d_j$-dimensional unitary operator on the subspace $\mathcal{H}_j$. 
We call $H_j$ {\it local quantum coin} at vertex $j$. 
Then we present two types of time evolutions of QWs, $U^{(G)}$ and $U^{(A)}$, respectively. 
\begin{definition}
( Gudder type and Ambainis type QWs. ) 
\begin{align}
U^{(G)} &= CS_\pi, \\
U^{(A)} &= S_\pi C. 
\end{align}
Here $S_\pi$ and $C$ are called shift and coin flip operators defined by 
\begin{align}
S_\pi |i,j\rangle &= |j,f_\pi(i,j) \rangle, \\
C &= \sum_{j\in V(\mathcal{G})} \oplus H_j,
\end{align}
that is 
\[ C|i,j\rangle=\sum_{k\in N(i)}\langle \bs{e}_k^{(i)}|H_i|\bs{e}_j^{(i)}\rangle|i,k\rangle. \]
\end{definition}
The first type determined by $U^{(G)}$ is a generalization of Gudder (1988) of $d$-dimensional lattice case. 
The second one $U^{(A)}$ is motivated by the most popular time evolution for the study of QWs by Ambainis et al (2001). 
We call such time evolution G-type QW and A-type QW, respectively. 
The matrix representations of $U_G$ and $U_A$ are as follows: for any $(i,j)$, $(l,m)\in \mathcal{D}(G)$, 
\begin{align}
\langle l,m|U^{(G)}|i,j\rangle &= \bs{1}_{\{l\in N(i)\}}\delta_{j,l} \langle \bs{e}^{(j)}_{m}|H_j|\bs{e}_{f_\pi(i,j)}^{(j)}\rangle, \label{Gm}\\
\langle l,m|U^{(A)}|i,j\rangle &= \bs{1}_{\{l\in N(i)\}}\delta_{m,f_\pi(i,l)}\langle\bs{e}_l^{(i)}|H_i|\bs{e}_j^{(i)}\rangle.
\end{align}
The dynamics of quantum walk is explained as follows. See also Fig. 2. Let us consider the canonical base $|i,j\rangle$ be acted 
by $U=SC$.  
In the coin flip stage $C$, the coin flip operator changes the terminal vertex $j$ to $l$ with the complex valued weight $(H_i)_{j,l}$. 
Thus in this stage, we obtain a superposition around the vertex $i$. 
In the next stage, that is, the shift $S$, the initial vertex $i$ is changed to its terminal vertex $l$, and the terminal vertex $l$ is changed to $\pi(j,l)$. 
This is the A type quantum walk. In G-type quantum walk, the order of shift and coin is just exchanged. 
\begin{figure}
\begin{center}
\includegraphics[width=10cm]{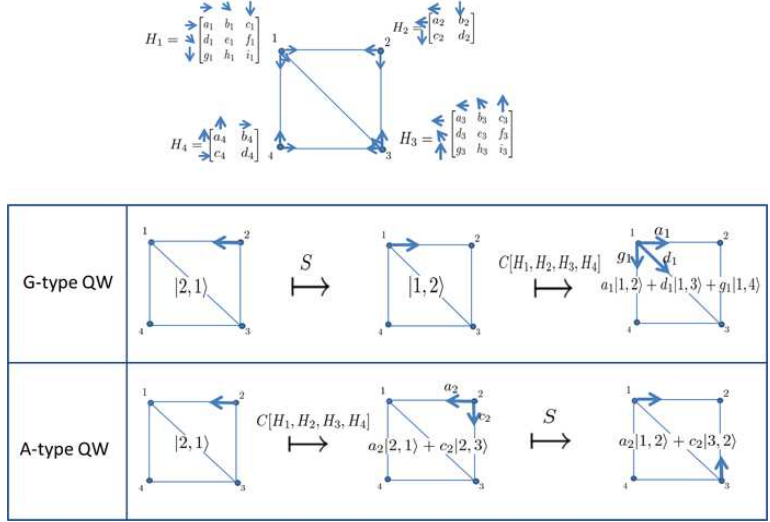}
\caption{ {\scriptsize Comparison between G-type and A-type QWs with flip flop $\pi_{ff}$:  We assign local quantum coins $H_j$ ($j=1,2,3,4$) 
which determines the weight of the ``pivot turn" at each vertex. 
The figure depicts the dynamics of G- and A- type QWs with $\pi=\pi_{ff}$ starting from the canonical base $|2,1\rangle$,  
that is, in G-type QW, $|2,1\rangle \stackrel{S}{\mapsto} |1,2\rangle \stackrel{C}{\mapsto} a_1|1,2\rangle+d_1|1,3\rangle+g_1|1,4\rangle$, 
on the other hand, in A-type QW, $|2,1\rangle \stackrel{C}{\mapsto} a_2|2,1\rangle+c_2|2,4\rangle \stackrel{S}{\mapsto} a_2|1,2\rangle+c_2|4,2\rangle$. 
}}
\end{center}
\end{figure}

\begin{remark}
The matrix valued weight $W_{(u,v)}$ associated with moving from $u$ to a neighbor $v$ in Definition~\ref{od} is as follows: 
\begin{equation}
W_{(u,v)}=
\begin{cases} 
	H_v|\bs{e}^{(v)}_{f_\pi(u,v)}\rangle\langle{\bs{e}^{(u)}_v}| & \text{: G-type, } \\  
        |\bs{e}^{(v)}_{f_\pi(u,v)}\rangle\langle{\bs{e}^{(u)}_v}| H_u & \text{: A-type. }
\end{cases}
\end{equation}
\end{remark}
A-type and G-type QWs are in dual relation with respect to the ``$1/2$" time gap: 
\begin{lemma}\label{dual}
For any $n\geq 0$, we have 
\begin{equation}
{U^{(G)}}^n=S_\pi^\dagger {U^{(A)}}^{n} S_\pi. 
\end{equation}
\end{lemma}
Because of the unitarity of the time evolution of quantum walks, ${U^{(J)}}^{-1}$ is also unitary. 
What is the ${U^{(J)}}^{-1}$ ? The following theorem is related to a part of its answer. 
\begin{lemma}\label{time_reverse}
${U^{(J)}}^{-1}$ is also a time evolution of a quantum walk $(J\in \{A,G\})$ on the same graph $\mathcal{G}(V,E)$ 
if and only if the shift operator of $U^{(J)}$ is the flip flop. 
More concretely, denote $U^{(J)}_{\pi_{ff}}[H_j:j\in V]$ as the time evolution of type $J$ $(J\in \{A,G\})$ 
quantum walk with local quantum coins $\{H_j\}_{j=1}^{|V|}$ and the flip flop shift. Then we have 
\begin{equation}
{U^{(J)}_{\pi_{ff}}[H_j:j\in V]}^{-1} =U^{(\neg J)}_{\pi_{ff}}[H_j^{-1}:j\in V]. 
\end{equation}
where $\neg J=A$ $(J=G)$, $=G$ $(J=A)$. 
\end{lemma}
In particular, if we choose local coins as self adjoint operators $H_j=H^\dagger_j$ such as the Grover coin $H_j=(2/d_j) J_{d_j}-I_{d_j}$ $(j\in V)$, 
\[ \left(U^{(J)}_{\pi_{ff}}\right)^{-1}=U^{(\neg J)}_{\pi_{ff}}. \]
where $J_m$ is the $m$-dimensional matrix whose elements are all one, and $I_m$ is the identity operator. 
\begin{proof}
Remark that
\begin{align} 
(U^{(G)})^{-1} &= (CS_\pi)^{-1}=S^{-1}_\pi C^{-1}. 
\end{align}
Note that $C^{-1}=\sum_{j=1}^{|V|}\oplus H_j^{-1}$ is also a coin flip operator. 
In the following, we concentrate on a necessary and sufficient condition for $\pi$ so that $S^{-1}_\pi$ is also a shift operator. 
For a partition $\pi\in \Pi_{\mathcal{G}}$ with $\pi: \overrightarrow{L}G\mapsto C_1 \oplus \cdots \oplus C_r$, we define $\pi^*$ as 
\begin{equation}\overrightarrow{L}G \mapsto C_1^{-1} \oplus \cdots \oplus C_r^{-1}. \end{equation}
Here for an essential cycle $C_k\subset \overrightarrow{L}G$, 
$(v_1,v_2)\to(v_2,v_3)\to\cdots\to(v_m,v_1)$, we define $C^{-1}$ as $(v_m,v_1)\to(v_{m-1},v_m)\to\cdots\to(v_1,v_2)$. 
Define $g_{\pi^*}: \bigcup_{j=1}^r V(C_j^{-1})\to V(\mathcal{G})$ such that 
\[\big((j,i),(i,g_{\pi^*}(j,i))\big)\in \bigcup_{j=1}^r A(C_j^{-1}). \]
%Remark that $V(\bigcup_{j=1}^{r}C_j^{-1})=V(\overrightarrow{L}(G))$ is always hold, 
%however $D(\bigcup_{j=1}^{r}C_j^{-1})=D(\overrightarrow{L}(G))$ dose not ensured in general. 
%Indeed, only if $\pi$ is the flip flop, that is, $\pi=\pi^*$, then $\pi^* \in \Pi_{\mathcal{G}}$. 
%Because $S_\pi$ is a permutation on the canonical basis $\{|i,j\rangle :(i,j)\in \mathcal{D}(G)\}$, 
%$S^{-1}_\pi$ is also a permutation, 
Then it is hold that for $(i,j)\in \bigcup_{j=1}^rV(C_j)$, 
\begin{equation}S^{-1}_\pi|i,j\rangle=|g_{\pi^*}(j,i),i\rangle. \end{equation}
Therefore $S^{-1}_\pi$ is a shift operator if and only if $g_{\pi^*}(j,i)=j$, that is, $\pi$ is the flip flop. 
\end{proof} 
\begin{lemma}\label{change_pi}
For any $\pi,\pi'\in \Pi_{\mathcal{G}}$, for each vertex $j\in V(\mathcal{G})$, there exists a permutation $\mathcal{P}^{(j)}_{\pi,\pi'}$ 
on the canonical basis of $\mathcal{H}_j$, 
$\{|\boldsymbol{e}_k^{(j)}\rangle: k\in N(j)\}$, such that 
\begin{equation}
U_{\pi'}^{(G)}[H_j: j\in V(\mathcal{G})]=U_{\pi}^{(G)}[\widetilde{H}_j: j\in V(\mathcal{G})],
\end{equation}
where $\widetilde{H}_j=H_j\mathcal{P}^{(j)}_{\pi,\pi'}$. 
\end{lemma}
\begin{proof}
Note that for any $j\in V(\mathcal{G})$, and $\pi,\pi'\in \Pi_{\mathcal{G}}$, 
\[N(j)=\{f_\pi(i,j);i\in N(j)\}=\{f_{\pi'}(i,j);i\in N(j)\}. \] 
Then we can define a permutation on $N(j)$ such that $\sigma_{\pi,\pi'}^{(j)}: f_\pi(i,j)\mapsto f_{\pi'}(i,j)$. Denote $\mathcal{P}_{\pi,\pi'}^{(j)}$ as 
the matrix representation of $\sigma_{\pi,\pi'}^{(j)}$ on $\mathcal{H}_j$, such that 
\begin{equation}\label{pipi'}
\mathcal{P}_{\pi,\pi'}^{(j)}=\sum_{i\in N(j)}| \bs{e}_{f_{\pi'}(i,j)}^{(j)} \rangle\langle \bs{e}_{f_\pi(i,j)}^{(j)} | \cong \sum_{i\in N(j)}| j,f_{\pi'}(i,j) \rangle\langle j,f_\pi(i,j) |.
\end{equation}
The permutation operator $\mathcal{P}_{\pi,\pi'}^{(j)}$ locally changes a partition $\pi\in \Pi_{\mathcal{G}}$ to another partition $\pi'\in \Pi_{\mathcal{G}}$ at vertex $j$. 
Combining Eq.~(\ref{pipi'}) with $S_\pi=\sum_{(i,j)}|j,f_\pi(i,j)\rangle\langle i,j|$ implies 
\[\sum_{j\in V(\mathcal{G})}\oplus \mathcal{P}_{\pi,\pi'}^{(j)}S_\pi=S_{\pi'}.\] 
So we have 
\begin{align}
U_{\pi'}^{(G)}[H_j: j\in V(\mathcal{G})] &= C S_{\pi'}=\sum_{j\in V(\mathcal{G})}\oplus H_j \cdot \sum_{i\in V(\mathcal{G})}\oplus \mathcal{P}_{\pi,\pi'}^{(i)}S_\pi \\
 &= \left(\sum_{j\in V(\mathcal{G})}\oplus H_j\mathcal{P}_{\pi,\pi'}^{(j)}\right)\cdot S_\pi \\
 &= U_{\pi}^{(G)}[\widetilde{H}_j: j\in V(\mathcal{G})], 
\end{align}
where $\widetilde{H}_j=H_j\mathcal{P}_{\pi,\pi'}^{(j)}$. It completes the proof. 
\end{proof}
\begin{theorem}\label{thm1}
Every G-type QW can be expressed by an A-type QW with flip flop shift $\pi_{ff}$ in the following meaning: 
for every $\pi\in\Pi_{\mathcal{G}}$, and a sequence of local quantum coins $\{H_j\}_{j=1}^{|V|}$, 
\begin{equation}\label{AtoG}
U_\pi^{(G)}[H_j:j\in V(\mathcal{G})]={U_{\pi_{ff}}^{(A)}}^{\dagger}[\widetilde{H}_j^\dagger: j\in V(\mathcal{G})], 
\end{equation}
where $\widetilde{H}_j=H_j\mathcal{P}_{\pi_{ff},\pi}^{(j)}$.
\end{theorem}
\begin{proof}
Combining Lemma \ref{time_reverse} with \ref{change_pi}, we arrive at
\begin{align}
U_{\pi}^{(G)}[H_j: j\in V(\mathcal{G})]=U_{\pi_{ff}}^{(G)}[\widetilde{H}_j: j\in V(\mathcal{G})]={U_{\pi_{ff}}^{(A)}}^{\dagger}[\widetilde{H}_j^\dagger: j\in V(\mathcal{G})]. 
\end{align}
\end{proof}
\begin{corollary}\label{AtoA}
For every $\pi\in \Pi_{\mathcal{G}}$, Ambainis type QW with $\pi$ and a sequence local quantum coins $\{H_j\}_{j=1}^{|V|}$, 
can be also expressed by an Ambainis type QW with the flip flop shift $\pi_{ff}$ as follows: 
\begin{equation}\label{AtoA1}
U_\pi^{(A)}[H_j:j\in V(\mathcal{G})] =S_\pi{U_{\pi_{ff}}^{(A)}}^{\dagger}[\widetilde{H}_j^\dagger:j\in V(\mathcal{G})]S_\pi^\dagger,
\end{equation}
where $\widetilde{H}_j=H_j\mathcal{P}_{\pi_{ff},\pi}^{(j)}$.
\end{corollary}
\begin{proof}
Lemmas \ref{dual} and \ref{change_pi} and Theorem \ref{thm1} imply that 
\begin{align}
U_\pi^{(A)}[H_j:j\in V(\mathcal{G})] &= S_\pi{U_{\pi}^{(G)}}[H_j:j\in V(\mathcal{G})]S_\pi^\dagger \\
	&= S_\pi{U_{\pi_{ff}}^{(G)}}[\widetilde{H}_j:j\in V(\mathcal{G})]S_\pi^\dagger \\
        &= S_\pi{U_{\pi_{ff}}^{(A)}}^{\dagger}[\widetilde{H}_j^\dagger:j\in V(\mathcal{G})]S_\pi^\dagger, 
\end{align}
which completes the proof. 
\end{proof}
%\begin{remark}
%Since ${U_{\pi_{ff}}^{(A)}}^\dagger=C^\dagger S^\dagger_{\pi_{ff}}=C^\dagger S_{\pi_{ff}}$, 
%we obtain for every $n\geq 1$, 
%\begin{equation}
%{U_\pi^{(A)}[H_j:j\in V(\mathcal{G})]}^n=U_\pi^{(A)}[\widetilde{H}_j^\dagger:j\in V(\mathcal{G})]\cdot {U_{\pi_{ff}}^{(A)}[\widetilde{H}_j^\dagger:j\in V(\mathcal{G})]}^{n-1}\cdot S_{\pi_{ff}}S_\pi^{\dagger}. 
%\end{equation}
%\end{remark}
%
For matrices $M, M'$, if there exists a permutation matrix $P$ such that $M'=P^\dagger MP$, we call $M$ is isomorphic to $M'$. 
\begin{corollary}(Severini \cite{Severini})
Every time evolution of coined QW is a weighted adjacency matrix of $\overrightarrow{L}\mathcal{G}$ or isomorphic to its transposed one. 
\end{corollary}
\begin{proof}
The adjacency matrix of $\overrightarrow{L}\mathcal{G}$ is 
\begin{equation}\label{Ad} \langle l,m|M(\overrightarrow{L}\mathcal{G})|i,j\rangle = \delta_{j,l}. \end{equation}
Comparing the Eq.~(\ref{Ad}) with Eq.~(\ref{Gm}), obviously, 
G-type QW is a weighted adjacency matrix of $\overrightarrow{L}\mathcal{G}$. 
Putting $J_m$ be $m$-dimensional all one matrix, we have for every $\pi\in \Pi_\mathcal{G}$, 
\[ M(\overrightarrow{L}\mathcal{G}) = \left(\sum_{j\in V}\oplus J_{d_j}\right) S_\pi. \] 
Therefore, for every $\pi\in \Pi_\mathcal{G}$, by the statement of proof for Theorem \ref{thm1}, 
\begin{align}
M(\overrightarrow{L}\mathcal{G})^\dagger &= \left\{\left(\sum_{j\in V}\oplus J_{d_j}\right) S_\pi\right\}^\dagger
	=\left\{\left(\sum_{j\in V}\oplus (J_{d_j}\mathcal{P}_{\pi,\pi_{ff}}^{(j)})\right) S_{\pi_{ff}}\right\}^\dagger \\
        &=\left\{\left(\sum_{j\in V}\oplus J_{d_j}\right) S_{\pi_{ff}}\right\}^\dagger \\
        &= S_{\pi_{ff}}\left(\sum_{j\in V}\oplus J_{d_j}\right), 
\end{align}
which implies that A-type QW with flip flop partition is a transposed weighted adjacency matrix of $\overrightarrow{L}\mathcal{G}$. 
Moreover from Corollary \ref{AtoA}, obviously, we see that A-type QW with partition $\pi\in \Pi_\mathcal{G}$ is isomorphic 
to a transposed weighted adjacency matrix of the line digraph of $\mathcal{G}$ 
with respect to the permutation matrix $S^\dagger _\pi$. 
So we obtain the desired conclusion. 
\end{proof}
For a fixed coin operator $C$, then 
once we get an information on the A-type QW with {\it flip flop shift}, we can immediately interpret it to any other corresponding coined quantum walk 
because of Eq.~(\ref{AtoG}) in Theorem \ref{thm1} and Eq.~(\ref{AtoA1}) in Corollary \ref{AtoA}. 
Thus from now on, we treat only A-type QWs with flip flop shift. 
Note that {\it all A-type QWs with flip flop shift on graph $\mathcal{G}$ are determined by only the choice of local quantum coins $H_j$'s} ($j\in V(\mathcal{G})$). 
In the following, we will show two special choices of the local quantum coins called ``Szegedy walk" and ``quantum graph walk". 
%To explain an important difference between types G and A, 
%let us consider what happens when each $U_J$ $(J\in \{A,G\})$ acts a canonical base $|i,j\rangle$:  
%In type G, at first, the state $|i,j\rangle$ changes its origin $i$ to the terminal vertex $j$ according to $f_\pi$. 
%this is the shift operator. Secondly, a quantum superposition at vertex $j$ is produced by the local coin operator at $H_j$ 
%from the moving state $|j,f_\pi(i,j) \rangle$. 
%On the other hand, in type A, at first a scattering state is given by $H_i$ from $|i,j\rangle$, 
%and then each state extending from $i$ moves to its neighbor accoding to $S_\pi$. 
%%%%%%%%%%%%%%%%%%%%%%%%%%%%%%%%%%%%%%%%%%%%%%%%%%%%%%%%%%%%%%%%%%%%%%%%%%%%%%%%%%%%%%%%%%%%%%%%%%%%%%%%%%%
%%%%%%%%%%%%%%%%%%%%%%%%%%%%%%%%%%%%%%%%%%%%%%%%%%%%%%%%%%%%%%%%%%%%%%%%%%%%%%%%%%%%%%%%%%%%%%%%%%%%%%%%%%%
\section{Szegedy walk}\label{sec:3}
%%%%%%%%%%%%%%%%%%%%%%%%%%%%%%%%%%%%%%%%%%%%%%%%%%%%%%%%%%%%%%%%%%%%%%%%%%%%%%%%%%%%%%%%%%%%%%%%%%%%%%%%%%%
%%%%%%%%%%%%%%%%%%%%%%%%%%%%%%%%%%%%%%%%%%%%%%%%%%%%%%%%%%%%%%%%%%%%%%%%%%%%%%%%%%%%%%%%%%%%%%%%%%%%%%%%%%%
In this section we briefly review on the Szegedy walk. 
The original walk introduced by Szegedy himself is the double steps of the Szegedy walk treated here. 
The Szegedy walk comes from a probability transition matrix $(\bs{P})_{u,v\in V(\mathcal{G})}$ on graph $\mathcal{G}$. 
Put $(\bs{P})_{u,v}=p_{u,v}$ which is the probability that a particle on vertex $u$ jumps to the neighbor $v$ at each time step with 
$\sum_{y\in N(u)}p_{u,v}=1, 0\leq p_{u,v}\leq 1$. 
\begin{definition} (Szegedy walk)
We call {\it Szegedy walk} to the A-type QW with flip flop shift $U^{(\bf{P})}_{\pi_{ff}}[H_j;j\in V]$, where the $d_j$-dimensional unitary local quantum coin at vertex $j$ is
for any $l,m\in N(j)$, 
\begin{equation}\label{SzeCoin}
\langle \bs{e}_m^{(j)}|H_j|\bs{e}_l^{(j)}\rangle=2\sqrt{p_{j,l}p_{j,m}}-\delta_{lm}.
\end{equation}
\end{definition}
\noindent Put $A: \ell^2(V)\to \ell^2(D)$ such that for a canonical base $|j\rangle$ $(j\in V)$, $A|j\rangle= \sum_{l\in N(j)}\sqrt{p_{j,l}}|j,l\rangle$. 
In particular, we choose $\bs{P}$ so that $p_{i,j}=1/d_i$ for all $i\in V$, the Szegedy walk becomes the Grover walk which is intensively investigated in the 
view point of quantum information. 
Let the symmetric matrix $J\in M_{|V|}(\mathcal{C}^2)$ be $(J)_{ij}=\sqrt{p_{ij}p_{ji}}$. 
In the Grover walk case, $J=\bs{P}$. 
Then we can obtain the eigensystem of $U^{(\bs{P})}$ by using the eigensystem of $J$ as follows. 
In this paper, we refine the original theorem by Szegedy~\cite{Szegedy}. (We can see for a detailed proof in \cite{Segawa} for example.) 
\begin{theorem}\label{szegedywalk}
Let $\nu=\cos\theta_\nu$ with $\mathrm{sgn}(\sin\theta_\nu)=\mathrm{sgn}(\nu)$. 
Then we have 
\begin{equation}
\mathrm{spec}(U^{(\bs{P})}) = 
\begin{cases}
	\{ e^{\im \theta_\nu}; \nu\in\mathrm{spec}(J) \}\cup \{ e^{-\im \theta_\nu}; \nu\in\mathrm{spec}(J) \setminus{\{\pm 1\}} \} & \text{; $|E|=|V|-1$}, \\
	\{ e^{\im \theta_\nu}; \nu\in\mathrm{spec}(J) \}\cup \{ e^{-\im \theta_\nu}; \nu\in\mathrm{spec}(J) \}& \text{; $|E|=|V|$}, \\
	\{ e^{\im \theta_\nu}; \nu\in\mathrm{spec}(J) \}\cup \{ e^{-\im \theta_\nu}; \nu\in\mathrm{spec}(J) \}\cup \{\overbrace{1}^{|E|-|V|},\overbrace{-1}^{|E|-|V|} \} & \text{; otherwise}. 
\end{cases}
\end{equation}
Let $\bs{\mathfrak{p}}_\nu$ the eigenvector of eigenvalue $\nu$ for $J$. 
The eigenvectors for 
\[ 
e^{\im \theta_\nu} \mathrm{\;with\;} \nu\in \mathrm{spec}(J) \mathrm{\;and\;} 
e^{-\im \theta_\nu} \mathrm{\;with\;} \nu\in \mathrm{spec}(J)\setminus\{\overbrace{1}^{m(1)}, \overbrace{-1}^{m(-1)}\} 
\]
are expressed by 
\begin{equation}
        (I-e^{i\theta_\nu} S)A\bs{\mathfrak{p}}_\nu \mathrm{\;and\;} (I-e^{-i\theta_\nu}S)A\bs{\mathfrak{p}}_\nu, 
\end{equation}
respectively, where $m(\pm 1)$ are the multiplicities of eigenvalues $\pm 1$ of $J$. 
\end{theorem}
%
%%%%%%%%%%%%%%%%%%%%%%%%%%%%%%%%%%%%%%%%%%%%%%%%%%%%%%%%%%%%%%%%%%%%%%%%%%%%%%%%%%%%%%%%%%%%%%%%%%%%%%%%%%%
%%%%%%%%%%%%%%%%%%%%%%%%%%%%%%%%%%%%%%%%%%%%%%%%%%%%%%%%%%%%%%%%%%%%%%%%%%%%%%%%%%%%%%%%%%%%%%%%%%%%%%%%%%%
\section{Quantum graph walk}\label{sec:4}
%%%%%%%%%%%%%%%%%%%%%%%%%%%%%%%%%%%%%%%%%%%%%%%%%%%%%%%%%%%%%%%%%%%%%%%%%%%%%%%%%%%%%%%%%%%%%%%%%%%%%%%%%%%
%%%%%%%%%%%%%%%%%%%%%%%%%%%%%%%%%%%%%%%%%%%%%%%%%%%%%%%%%%%%%%%%%%%%%%%%%%%%%%%%%%%%%%%%%%%%%%%%%%%%%%%%%%%
\subsection{Quantum graphs}
This formulation of the quantum graph is according to Smilansky and his group \cite{Smilansky}. 
In the quantum graph, a metric graph of $\mathcal{G}(V,E)$, whose each edge $e \in E(\mathcal{G})$ is assigned a length $L_{e}\in [0,\infty)$, is given.  
Let us denote the vertex set $V(\mathcal{G})$ which has an order such that $V=\{1,2,\dots, |V|\}$. 
To describe position on edge $e=\{i,j\}$ of the metric graph $\mathcal{G}(V,E)$, we define $x\in [0,L_e]$ by the distance from $\mathrm{min}\{i,j\}$. 

At each edge $\{i,j\}\in E(\mathcal{G})$, the quantum graph gives the wave function $\Psi_{\{i,j\}}(x)$ in the location of $x\in [0,L_{\{i,j\}}]$ 
determined by the following Schr\"odinger equation: 
\begin{equation}\label{original}
\left( -\im \frac{d}{dx}+A_{\{i,j\}} \right)^2 \Psi_{\{i,j\}}(x)=k^2 \Psi_{\{i,j\}}(x). 
\end{equation}
Moreover the wave function is imposed the following two boundary conditions: 
\begin{enumerate}
\item\label{Continuty} Continuity \\
For every $i\in V(\mathcal{G})$, there exists a $\phi_i\in \mathbb{C}$, such that
\begin{align}
\Psi_{\{i,j\}}(0) &= \phi_i\mathrm{\;\;for\;any\;} j\in N(i) \mathrm{\;with\;} j>i,  \\
\Psi_{\{i,k\}}(L_{\{i,k\}}) &= \phi_i\mathrm{\;\;for\;any\;} k\in N(i) \mathrm{\;with\;} k<i.
\end{align}
where $N(i)=\{j\in V(\mathcal{G}): \{i,j\}\in E(\mathcal{G})\}$. 
\item \label{Current conservation}Current conservation\\ For $\lambda_i\geq 0$, 
\begin{multline}
	\sum_{j:j<i}\left( -\frac{d}{dx}-\im A_{\{ij\}} \right)\Psi_{\{ij\}}(x)\bigg|_{x=L_{ij}}
	+\sum_{j:j>i}\left( -\frac{d}{dx}+\im A_{\{ij\}} \right)\Psi_{\{ij\}}(x)\bigg|_{x=0}=\lambda_i\phi_i. 
\end{multline}        
\end{enumerate}
When $\lambda_i=0$, then the condition \ref{Current conservation} is called Neumann boundary condition, while $\lambda_i=\infty$, Dirichlet boundary condition.
Define the following wave function on $\mathcal{D}(\mathcal{G})$: 
\begin{equation}\label{change_arc}
\Psi_{(i,j)}(x) = \begin{cases}  \Psi_{\{ij\}}(x) & \text{: $i<j$,} \\ \Psi_{\{ij\}}(L_{\{ij\}}-x)  & \text{: $i>j$}\end{cases}.
\end{equation}
Let $A_{(ij)}=\mathrm{sgn}(j-i)A_{\{ij\}}$. 
Then we obtain the following lemma which is equivalent to the original Schr\"odinger equation (\ref{original}) with the two boundary conditions (1) and (2), however it is useful for our discussion: 
\begin{lemma}
The Schr\"odinger equations (\ref{original}) with the boundary conditions (1) and (2) are hold for all $\{ij\}\in E$ simultaneously, if and only if 
the following Schr\"odinger equations (\ref{arcSch}) with the boundary conditions (\ref{c1}) - (\ref{c3}) are hold for all $(i,j)\in D(\mathcal{G})$. 
\begin{equation}\label{arcSch}
\left(-\im\frac{d}{dx}+A_{(i,j)}\right)^2\Psi_{(i,j)}(x) =k^2 \Psi_{(i,j)}(x).
\end{equation}
\begin{enumerate}
\renewcommand{\theenumi}{\Roman{enumi}}
\item\label{c1} $\Psi_{(i,j)}(x)=\Psi_{(j,i)}(L_{\{i,j\}}-x)$,
\item\label{c2} $\Psi_{(i,j)}(0)=\phi_i$ for all $j\in N(i)$. 
\item\label{c3} $\sum_{j\in N(i)}\left( -\im d/dx +A_{(i,j)}\right)\Psi_{(i,j)}(x)\big|_{x=0}=-\im\lambda_i\phi_i$ for all $i\in V(\mathcal{G})$. 
\end{enumerate}
\end{lemma}
%%%%%%%%%%%%%%%%%%%%%%%%%%%%%%%%%%%%%%%%%%%%%%%%%%%%%%%%%%%%%%%%%%%%%%%%%%%%%%%%%%%%%%%%%%%%%%%%%%%%%%%%%%%%%%%%%%%%%%%%%%%%%%%%%%%%%%%%%%%%%%%%%%%%
\subsection{Quantum graph walk}
We should note that the quantum graph is determined by sequence of edge length $\bs{L}=\{L_{\{ij\}};\{ij\}\in E\}$, and 
boundary conditions at each vertex $\bs{\lambda}=\{\lambda_j;j\in V\}$ and the vector potential with respect to magnetic flux $\bs{A}=\{A_{\{ij\}};\{ij\}\in E\}$.  
\begin{definition}(Quantum graph walk)
We call quantum graph walk with parameters of quantum graph $(\bs{L}, \bs{\lambda}, \bs{A})$ to the A-type QW with flip flop shift 
\[ U^{(\bs{L},\bs{\lambda},\bs{A})}(k)\equiv U^{(A)}_{\pi_{ff}}[H_j(k);j\in V(\mathcal{G})], \]
where 
\begin{equation}\label{m_element}
\langle \bs{e}^{(j)}_m|H_j(k)|\bs{e}^{(j)}_l\rangle= \left(\frac{2}{d_j+\im \lambda_j/k}-\delta_{l,m} \right)e^{\im L_{\{jm\}}(k-A_{(jm)})}. 
\end{equation}
\end{definition}
\begin{remark}
An equivalent expression for $H_j(k)$ is 
\[ H_j(k)
=D_j(k)\left( \frac{2}{d_j+\im \lambda_j/k}J_{d_j}-I_{d_j} \right). \]
where $D_j(k)$ is a diagonal matrix such that $\sum_{m\in N(j)} e^{\im L_{\{jm\}}(k-A_{(jm)})}|\bs{e}^{(j)}_m \rangle\langle \bs{e}^{(j)}_m|$, 
and $J_{d_j}$ is the all $1$ matrix, $I_{d_j}$ is the identity matrix on $\mathcal{H}_j$. 
\end{remark}
\begin{remark}
In the limit of $\bs{L}\downarrow \bs{0}$ with the Neumann boundary condition, the Grover walk appears again. 
Comparing both expressions for the local quantum coins for the Szegedy walk (Eq.~(\ref{SzeCoin})) and quantum graph walk (Eq.~(\ref{m_element})), 
the common class of both walks is only the Grover walk. 
\end{remark}
A general solution for Eq.~(\ref{arcSch}) can be directly solved by using two parameters $a_{(i,j)}, b_{(i,j)}\in \mathbb{C}$, 
\begin{equation}\label{generalSol}
\Psi_{(i,j)}(x)=\left( a_{(i,j)}e^{-\im kx}+b_{(i,j)}e^{\im kx} \right)e^{-\im A_{(i,j)}x}.
\end{equation}
\begin{lemma}\label{b_and_a}
It is hold that 
\begin{equation}\label{b_and_a_eq}
b_{(i,j)}=a_{(j,i)}e^{-\im L_{\{ij\}}(k-A_{(i,j)})}.
\end{equation}
\end{lemma}
\begin{proof}
Substituting Eq.~(\ref{generalSol}) into the condition (\ref{c1}), it is hold that for any $(i,j)\in \mathcal{D}(G)$ and $x\in [0,L_{\{ij\}}]$, 
\begin{equation}\label{identity}
a_{(i,j)}e^{-\im kx}+b_{(i,j)}e^{\im kx}=\left\{ a_{(j,i)}e^{-\im L_{\{ij\}}(k-A_{(i,j)})} \right\} e^{\im kx}
	+\left\{ b_{(j,i)}e^{\im L_{\{ij\}}(k+A_{(i,j)})} \right\} e^{-\im kx}. 
\end{equation}
Thus comparing the coefficients of $e^{-\im kx}$ and $e^{\im kx}$ of LHS with ones of RHS in the identity~(\ref{identity}) with respect to $x\in [0,L_{ij}]$, 
we obtain 
\begin{align}
a_{(i,j)} &= b_{(j,i)}e^{\im L_{\{ij\}}(k+A_{(i,j)})}, \label{identity1} \\
b_{(i,j)} &= a_{(j,i)}e^{-\im L_{\{ij\}}(k-A_{(i,j)})}. \label{identity2} 
\end{align}
Remarking that $A_{(j,i)}=-A_{(i,j)}$, then Eq.~(\ref{identity1}) is equivalent to Eq.~(\ref{identity}), 
we complete the proof. 
\end{proof}
By substituting Eq.~(\ref{b_and_a_eq}) into Eq.~(\ref{generalSol}), we obtain for each $(ij)\in D$, 
\begin{equation}\label{pp}
\Psi_{(ij)}(x)=a_{(ij)}e^{-\im(k+A_{(ij)})x}+a_{(ji)}e^{-\im(k+A_{(ji)})(L_{\{ij\}}-x)}. 
\end{equation}
Therefore $|D|$-parameter $\{a_{f};f\in D\}$ gives the solution for the Schr{\"o}dinger equations. 
We put $\bs{a}_*(k)$ as the array $a_{(ij)}$'s, that is, $\bs{a}_*(k)=\sum_{(ij)\in D}a_{(ij)}|i,j\rangle$. 
On the other hand, for $\bs{x}=(x_{(ij)}; (ij)\in D \mathrm{\;with\;}0\leq x_{ij}\leq L_{\{ij\}})$, and $k\in \mathbb{R}$, 
let the array of eigenfunctions $\Psi_{(ij)}(x_{(ij)})$'s be 
$\bs{\Psi}_*(k,\bs{x}) \equiv \sum_{i,j\in D}\Psi_{(i,j)}(x_{i,j})|i,j\rangle$. 
Then Eq.~(\ref{pp}) implies that 
\begin{equation} \label{wq}
\bs{\Psi}_*(k,\bs{x})=\left\{D_1(k,\bs{x})+D_2(k,\bs{x})S\right\}\bs{a}_*(k), 
\end{equation}
where $D_j(k,\bs{x})$ $(j\in \{1,2\})$ are diagonal matrix defined by for $f,f'\in D(\mathcal{G})$, 
\begin{align*}
(D_1)_{f,f'} &=  \delta_{f,f'}e^{-\im (k+A_f)x_f}, \\
(D_2)_{f,f'} &=  \delta_{f,f'}e^{-\im (k-A_f)(L_{f}-x_f)}. 
\end{align*}
%So once we get the $|D|$ parameters $a_{(ij)}$s, then we obtain the eigenfunction of the quantum graph by Eq.~(\ref{wq}). 
Now we will investigate a necessary and sufficient condition of $\bs{a}_*(k)$ for getting 
non-trivial solution of quantum graph $\bs{\Psi}_*(k,\bs{x})$ ($\neq \bs{0}$). 
One of its answers is our main result in Theorem \ref{mainthm}. 
The following theorem is a collection of equivalent statements including Theorem \ref{mainthm}. 
\begin{theorem}\label{QG}
The following three statements are equivalent: 
\begin{enumerate}
\item\label{non_tri_sol} 
In the quantum graph with parameters $(\bs{L},\bs{\lambda},\bs{A})$, 
the Schr\"odinger equation (\ref{arcSch}) with the boundary conditions (\ref{c1}) - (\ref{c3}) has a non-trivial solution 
$\{\Psi_{(i,j)}(x)\}_{(i,j)\in \mathcal{D}(G)}$. 
%\item\label{Gud} The one-step time evolution of the G-type QW with the flip flop shift and the local quantum coins $\{H_j(k)\}_{j=1}^{|V|}$, $U^{(G)}(k)$, has the eigenvalue $1$. 
\item\label{Amb} $\bs{a}_*(k)$ is an eigenvector of the quantum graph walk $U^{(\bs{L},\bs{\lambda},\bs{A})}(k)$ with eigenvalue $1$. 
\item\label{Sat} It is hold that
\begin{equation}
\mathrm{det}(I_{|V|}-T_{|V|}+D_{|V|})\prod_{j=1}^{|E|} (1-e^{2\im kL_{e_j}})=0, 
\end{equation}
where for $i,j\in V(\mathcal{G})$, 
\begin{align}
\left(T_{|V|}\right)_{i,j} &= \frac{e^{-\im L_{\{ij\}}(k+A_{(i,j)})}(1+e^{-\im \rho_j(k)})/\sqrt{d_id_j}}{1-e^{2\im kL_{\{ ij\}}}} \mathbf{1}_{\{(i,j)\in \mathcal{D}(G)\}}(i,j), \\
\left(D_{|V|}\right)_{i,j} &= \sum_{l\in N(i)} \frac{e^{2\im kL_{\{il\}}}(1+e^{-\im \rho_i(k)})/d_i}{1-e^{2\im kL_{\{il\}}}}\mathbf{1}_{\{i=j\}}(i,j).
\end{align}
\end{enumerate}
Here $e^{\im \rho_j(k)}=\{1+\im\lambda_j/(kd_j)\}/\{1-\im\lambda_j/(kd_j)\}$. 
\end{theorem}
\begin{proof}
At first we give the following lemma. 
\begin{lemma}\label{b_and_a_2}
The boundary conditions (\ref{c1})-(\ref{c3}) are hold 
for all $(i,j)\in \mathcal{D}(G)$, 
\begin{align}
\Leftrightarrow \;\; a_{(ij)} &= \sum_{l\in N(i)}\left( \frac{2}{d_i-\im \lambda_i/k}-\delta_{lj} \right)e^{-\im L_{\{il\}}(k-A_{(il)})}a_{(li)}  \label{ab1}
%\Leftrightarrow \;\; b_{(i,j)} &= \sum_{l\in N(i)}\left( \frac{2}{d_i+\im \lambda_i/k}-\delta_{lj} \right)a_{(i,l)}  \label{ab11}
\end{align}
\end{lemma}
\begin{proof}
We assume that the boundary conditions (\ref{c2}) and (\ref{c3}) are hold. 
From condition (\ref{c2}), substituting $x=0$ into Eq.~(\ref{generalSol}), 
\begin{equation}\label{from_c2}
\Psi_{(i,j)}(0)=a_{(i,j)}+b_{(i,j)}=\phi_i,\;\;j\in N(i).
\end{equation}
Taking a summation of Eq.~(\ref{from_c2}) over all the neighbors of $i$, 
\begin{equation}\label{from_c2_02}
\sum_{j\in N(i)}\left(a_{(i,j)}+b_{(i,j)}\right)=d_i\phi_i. 
\end{equation}
From Eq.~(\ref{generalSol}), 
\[ \frac{d}{dx} \Psi_{(i,j)}(x)\bigg|_{x=0}=-\im(k+A_{(i,j)})a_{(i,j)}+\im(k-A_{(i,j)})b_{(i,j)}, \] 
Inserting it into condition (\ref{c3}), we obtain 
\begin{equation}\label{from_c3}
-\im k \sum_{j\in N(i)} (a_{(i,j)}-b_{(i,j)})=\lambda_i\phi_i. 
\end{equation}
Combining Eq.~(\ref{from_c2_02}) with Eq.~(\ref{from_c3}), 
\[ \phi_i=-\frac{\im k}{\lambda_i} \sum_{j\in N(i)}(a_{(i,j)}-b_{(i,j)})=\frac{1}{d_i} \sum_{j\in N(i)}(a_{(i,j)}+b_{(i,j)}), \]
which implies that 
\begin{equation}\label{from_c2_and_c3}
\sum_{j\in N(i)}a_{(i,j)}=e^{\im \rho_i(k)}\sum_{j\in N(i)}b_{(i,j)}. 
\end{equation}
By using Eqs.~(\ref{from_c2}) (\ref{from_c2_02}) and (\ref{from_c2_and_c3}), 
\begin{align}
a_{(i,j)} &= \phi_i-b_{(i,j)}
	   = \frac{1}{d_i}\sum_{l\in N(i)} \left( a_{(i,l)}+b_{(i,l)} \right)-b_{(i,j)}, \notag \\
          &= \sum_{l\in N(i)}\left( \frac{2}{d_i-\im \lambda_i/k}-\delta_{l,j} \right)b_{(i,l)}. \label{hasegawan}
\end{align}
Conversely, under the assumption that Eq.~(\ref{hasegawan}) is hold, 
we can easily check that the conditions (\ref{c2}) and (\ref{c3}) are satisfied. 
Then inserting Lemma \ref{b_and_a} into Eq.~(\ref{hasegawan}), we complete the proof. 
\end{proof}
%
%
%From Lemma.~\ref{time_reverse}, we obtain 
%$U^{(G)}(k)={U^{(A)}}^{-1}(k)$ with the flip flop shift.  
%Thus
%\[ \mathrm{det}(I-U^{(A)}(k))=-\mathrm{det}(U^{(A)}(k))\mathrm{det}(I-U^{(G)}(k)), \]
%which implies that (\ref{Amb}) is equivalent to (\ref{Gud}). 
Next, we will give a proof that (\ref{non_tri_sol}) iff (\ref{Amb}). 
By using a matrix representation of the quantum coin at vertex $i$ in Eq.~(\ref{m_element}), RHS of Eq.~(\ref{ab1}) is rewritten by 
\[ \sum_{l\in N(i)}\langle \bs{e}^{(i)}_j|H_i^\dagger(k)| \bs{e}^{(i)}_l \rangle a_{(l,i)}, \]
which implies that $\bs{a}_*(k)=C^\dagger(k) S_{\pi_{ff}}\bs{a}_*(k)$ with $C(k)=\sum_{j\in V(\mathcal{G})}\oplus H_j(k)$. 
Note that from Lemma~\ref{time_reverse} the time reverse of the quantum graph walk is the following G-type quantum walk
\begin{equation}
\left(U^{(\bs{\bs{L}},\bs{\lambda},\bs{A})}\right)^{-1}=U^{(G)}_{\pi_{ff}}[H_j^\dagger (k);j\in V]. 
\end{equation}
Thus $\bs{a}_*(k)$ is the eigenvector of eigenvalue $1$ for both $U_{\pi_{ff}}^{(G)}[H_j^\dagger(k);k\in V]$ and 
$U^{(\bs{\bs{L}},\bs{\lambda},\bs{A})}\equiv U_{\pi_{ff}}^{(A)}[H_j (k);j\in V]$. 
Finally, we show that (\ref{Amb}) iff (\ref{Sat}). To do so, we give the following lemma: 
When we take $\alpha_{jl}=1/\sqrt{d_j}$ ($l\in N(j)$) and $t=1$ in the following lemma, then we obtain the statement of (\ref{Sat})
\begin{lemma}\label{sato_lemma}
Let $\widetilde{U}^{(A)}(k)$ be a generalized quantum graph walk whose quantum coin is denoted by 
\[ H_j (k)=D_j (k)\left\{(1+e^{-\im \rho_j(k)})\Pi_j-I_{d_j}\right\},\;\;\;(j\in V(\mathcal{G})),  \]
where $\Pi_j$ is a projection onto a unit vector 
$|\bs{\alpha}_j\rangle=\sum_{l\in N(j)}\alpha_{jl}|\bs{e}^{(j)}_l\rangle \in \mathcal{H}_j$ with $\sum_{l\in N(j)}|\alpha_{jl}|^2=1$. 
Then we have 
\begin{equation}
\mathrm{det}\left(I_{2|E|}-t\widetilde{U}^{(A)}(k)\right)
	=\mathrm{det}\left(I_{|V|}-tT_{|V|}(t)+t^2D_{|V|}(t)\right)\prod_{\{ij\}\in E} (1-t^2e^{-2\im kL_{\{ij\}}})
\end{equation}
where 
\begin{align}
\left(T_{|V|}(t)\right)_{i,j} &= \frac{e^{\im L_{\{ij\}}(k+A_{(i,j)})}(1+e^{-\im \rho_j(k)})\alpha_{ji}\overline{\alpha_{ij}}}{1-t^2e^{2ikL_{\{ ij\}}}} \mathbf{1}_{\{(i,j)\in \mathcal{D}(G)\}}(i,j), \\
\left(D_{|V|}(t)\right)_{i,j} &= \sum_{l\in N(i)} \frac{e^{2ikL_{\{il\}}}(1+e^{-\im \rho_i(k)})|\alpha_{ij}|^2}{1-t^2 e^{2ikL_{\{il\}}}}\mathbf{1}_{\{i=j\}}(i,j)
\end{align}
\end{lemma}
\begin{remark}
If we choose the unit vector $|\bs{\alpha}_j\rangle$ on each $\mathcal{H}_j$ as $|\bs{\alpha}_j\rangle=1/\sqrt{d_j}\sum_{l\in N(j)}|\bs{e}^{(j)}_l\rangle$, then the walk becomes a quantum graph walk. 
On the other hand, if we put the parameters $\bs{\lambda}=\bs{0}$, $\bs{L}=\bs{0}$, and $\alpha_{ij}\in [0,1]$ for all $(i,j)\in \mathcal{D}$, 
then the walk becomes a Szegedy walk. 
\end{remark}
In the following, we prove Lemma~\ref{sato_lemma}. 
For a sequence $(c_{(i,j)})_{(i,j)\in D(\mathcal{G})}$ and a sequence $(c_i)_{i\in V(\mathcal{G})}$, 
we denote $\mathcal{D}_{D}[(c_{(i,j)})_{(i,j)\in D(\mathcal{G})}]$ and $\mathcal{D}_{V}[(c_i)_{i\in V(\mathcal{G})}]$ as 
the following diagonal matrices on $\ell^2(D)$ and $\ell^2(V)$, respectively; 
\begin{align*}
\mathcal{D}_{D}[(c_{(i,j)})_{(i,j)\in D(\mathcal{G})}]
	= \sum_{(i,j)\in D(\mathcal{G})}c_{(i,j)}|i,j\rangle \langle i,j|, \;\;
\mathcal{D}_{V}[(c_i)_{i\in V(\mathcal{G})}]
	= \sum_{i\in V(\mathcal{G})}c_i|i\rangle \langle i|. 
\end{align*}
We will use the relation
\begin{equation}\label{changeorder}
S\mathcal{D}_{D}[(c_{(i,j)})_{(i,j)\in D(\mathcal{G})}]=\mathcal{D}_{D}[(c_{(j,i)})_{(i,j)\in D(\mathcal{G})}]
\end{equation}
Let $A$ as a matrix representation of a map $\ell^2(V) \to \ell^2(D)$ such that $i \mapsto |a_i\rangle$ for every $i\in V$, that is, 
$A = \sum_{j \in V}|a_j\rangle \langle j|$.  
Put 
\begin{align}\label{Diag}
B = S\cdot \widetilde{\mathcal{D}}_D \cdot A \cdot \widetilde{\mathcal{D}}_V, 
\end{align}
where $\widetilde{\mathcal{D}}_D= \mathcal{D}_D\left[\exp[\im L_{\{ij\}}(k-A_{(ij)})]: (ij)\in D\right]$, 
and $\widetilde{\mathcal{D}}_V=\mathcal{D}_V[1+e^{-\im \rho_j(k)}: j\in V ]$. 
The coin operator on $\ell^2(D)$ is described by 
\begin{equation}
C=\widetilde{\mathcal{D}}_D\left(A\widetilde{\mathcal{D}}_V A^\dagger-I_{|V|}\right).
\end{equation}
By using this, 
\begin{align}\label{eq1}
\mathrm{det}(I_{2|E|}-tU^{(A)}(k))
	&= \mathrm{det}\left(I_{2|E|}-tS\widetilde{\mathcal{D}}_D(A\widetilde{\mathcal{D}}_V A^\dagger-I_{|V|})\right) \notag \\
        &= \mathrm{det}(I_{2|E|}+tS\widetilde{\mathcal{D}}_D)\cdot \mathrm{det}\left(I_{2|E|}-t(I_{2|E|}+tS\widetilde{\mathcal{D}}_D)^{-1}B A^\dagger\right) \notag \\
        &= \mathrm{det}(I_{2|E|}+tS\widetilde{\mathcal{D}}_D)\cdot \mathrm{det}\left(I_{|V|}-t A^\dagger (I_{2|E|}+tS\widetilde{\mathcal{D}}_D)^{-1}B \right) .  
\end{align}
We should note that 
\begin{equation}
I_{2|E|}+tS\widetilde{\mathcal{D}}_D \cong \sum_{\{ij\}\in E} \oplus
	\begin{bmatrix} 1 & te^{\im L_{\{ij\}}(k-A_{(ji)})} \\ te^{\im L_{\{ij\}}(k-A_{(ij)})} & 1 \end{bmatrix}
\end{equation}
Put $\Delta_{\{ij\}}(t)=1-t^2e^{2\im kL_{\{ij\}}}$. 
Then we have 
\begin{align}
\mathrm{det}(I_{2|E|}+tS\mathcal{D}_D) &= \prod_{\{ij\}}\Delta_{\{ij\}}(t). \\
\left(I_{2|E|}+tS\widetilde{\mathcal{D}}_D \right)^{-1} 
	&= \widetilde{\mathcal{D}}_D^{(1)}\left(I-t\widetilde{\mathcal{D}}_D^{(2)}S\right), \label{popo}
\end{align}
where 
\[ \widetilde{\mathcal{D}}_D^{(1)} = \widetilde{\mathcal{D}}_D[\Delta_{\{ij\}}^{-1}(t);(ij)\in D],\;\; 
\widetilde{\mathcal{D}}_D^{(2)} = \widetilde{\mathcal{D}}_D[e^{\im L_{\{ji\}}(k-A_{(ji)})};(ij)\in D]. \]
We applied Eq.~(\ref{changeorder}) to the expression of Eq.~(\ref{popo}). 
By using these notations we rewrite $A^\dagger (I_{2|E|}+tS\widetilde{\mathcal{D}}_D)^{-1}B$ in Eq.~(\ref{eq1}) by
\[ A^\dagger (I_{2|E|}+tS\widetilde{\mathcal{D}}_D)^{-1}B
	= A^\dagger \widetilde{\mathcal{D}}_D^{(1)} B-tA^\dagger \widetilde{\mathcal{D}}_D^{(1)} \widetilde{\mathcal{D}}_D^{(2)} SB. \]
We can express the the first and second terms as 
\begin{align} 
A^\dagger \widetilde{\mathcal{D}}_D^{(1)} B
	&= A^\dagger \widetilde{\mathcal{D}}_D^{(1)} S \widetilde{\mathcal{D}}_D A \widetilde{\mathcal{D}}_V 
         = A^\dagger \widetilde{\mathcal{D}}_D\left[\left(e^{\im L_{\{ij\}}(k-A_{(ji)})}\Delta_{\{ij\}}^{-1}(t)\right)_{(ij)\in D}\right] SA \widetilde{\mathcal{D}}_V \notag \\
	&=\sum_{(ij)\in D} \overline{\alpha_{ij}}\cdot e^{\im L_{\{ij\}}(k-A_{(ji)})}\Delta_{\{ij\}}^{-1}(t) \cdot \alpha_{ji}\cdot (1+e^{-\im\rho_j(k)})
        	|i\rangle \langle j| \notag  \\
        &= T_{|V|}(t). \\
A^\dagger \widetilde{\mathcal{D}}_D^{(1)}\widetilde{\mathcal{D}}_D^{(2)} SB 
	&= A^\dagger \widetilde{\mathcal{D}}_D^{(1)}\widetilde{\mathcal{D}}_D^{(2)} S\cdot S \widetilde{\mathcal{D}}_D A \widetilde{\mathcal{D}}_V 
	 = A^\dagger \widetilde{\mathcal{D}}_D^{(1)}\widetilde{\mathcal{D}}_D^{(2)}\widetilde{\mathcal{D}}_D A \widetilde{\mathcal{D}}_V \notag \\
        &= A^\dagger \widetilde{\mathcal{D}}_D\left[e^{2\im kL_{\{ij\}}}/\Delta_{\{ij\}}:(ij)\in D\right] A \widetilde{\mathcal{D}}_V \notag \\
        &= \sum_{j\in V} \left(\sum_{i\in N(j)} (1+e^{\im \rho_j(k)}) |\alpha_{ij}|^2 e^{2\im kL_{\{ij\}}}/\Delta_{\{ij\}} \right) |j\rangle \langle j| \notag \\
        &= D_{|V|}(t). 
\end{align}
Then we complete the proof of Theorem \ref{QG}. 
\end{proof}
\subsection{ Necessary and sufficient conditions for quantum graph}
Finally, we mention the relation between quantum walk and quantum evolution map defined by \cite{Smilansky,Smilansky2}. 
In this paper, 
we have defined the A-type QW, $U^{(\bs{L,\lambda,A})}(k) \equiv U^{(A)}_{\pi_{ff}}[H_j(k);j\in V]$ with local quantum coins determined by 
the parameters of corresponding quantum graph $(\bs{L,\lambda,A})$ (see Eq.~(\ref{m_element})), as quantum graph walk. 
Recall that the statement of (2) in Theorem \ref{QG} is 
\begin{equation}\label{equivalent} U^{(\bs{L,\lambda,A})}\bs{a}_*(k)=\bs{a}_*(k) \end{equation}
which is an equivalent expression for satisfying the corresponding quantum graph. 
Since 
\[ U^{(A)}_{\pi_{ff}}[H_j(k);j\in V]=S_{\pi_{ff}}U^{(G)}_{\pi_{ff}}[H_j(k);j\in V]S_{\pi_{ff}}\;\;\mathrm{and}\;\;S_{\pi_{ff}}^2=I, \]
Eq.~(\ref{equivalent}) is reexpressed by   
\begin{equation}\label{GS} 
U^{(G)}_{\pi_{ff}}[H_j(k);j\in V]\bs{b}_*(k)=\bs{b}_*(k), 
\end{equation} 
where $\bs{b}_*(k)=S_{\pi_{ff}}\bs{a}_*(k)$. 
Combining Lemma \ref{time_reverse} with Eq.~(\ref{GS}), we can give equivalent statements to (1) in Theorem \ref{QG} as follows:
\begin{proposition}
The following statements are necessary and sufficient conditions for satisfying quantum graph 
\begin{align*}
U^{(A)}_{\pi_{ff}}[H_j(k);j\in V]\bs{a}_*(k)=\bs{a}_*(k) &\Leftrightarrow 
U^{(G)}_{\pi_{ff}}[H_j(k)^\dagger;j\in V]\bs{a}_*(k)=\bs{a}_*(k) \\
\Leftrightarrow U^{(A)}_{\pi_{ff}}[H_j(k)^\dagger;j\in V]\bs{b}_*(k) = \bs{b}_*(k) &\Leftrightarrow
U^{(G)}_{\pi_{ff}}[H_j(k);j\in V]\bs{b}_*(k) = \bs{b}_*(k).
\end{align*} 
\end{proposition}

The G-type QW, $U^{(G)}_{\pi_{ff}}[H_j(k);j\in V]$, is nothing but the ``quantum evolution map" in \cite{Smilansky,Smilansky2}. 
More concretely, the quantum evolution map is denoted by $\mathcal{U}_B(k)\equiv U^{(G)}_{\pi_{ff}}[H_j(k);j\in V]=\mathcal{T}(k)\mathcal{S}(k)$,
where $\mathcal{T}(k)$ and $\mathcal{S}(k)$ are called bond propagation matrix, and graph scattering matrix in their paper, respectively. 
The correspondence between the Simlansky's quantum evolution map and the G-type QW as follows: 
\begin{equation}
	\mathcal{T}(k)=C[\sigma_j;j\in V]S,\;\;\mathcal{S}(k)=\widetilde{\mathcal{D}}_D
\end{equation}
where for $l,m\in N(j)$, 
\[ \langle \bs{e}^{(j)}_m|\sigma_j|\bs{e}^{(j)}_l\rangle= \frac{2}{d_j+\im \lambda_j/k}-\delta_{l,m}, \]
and $\widetilde{\mathcal{D}}_D$ is defined in Eq.~(\ref{Diag}). 

We will be able to see more detailed discussions around here and new insight into quantum walks through the quantum graphs 
in our next papers \cite{HKSS_II,HKSS_III}. 
%%%%%%%%%%%%%%%%%%%%%%%%%%%%%%%%%%%%%%%%%%%%%%%%%%%%%%%%%%%%%%%%%%%%%%%%%%%%%%%%%%%%%%%%%%
%%%%%%%%%%%%%%%%%%%%%%%%%%%%%%%%%%%%%%%%%%%%%%%%%%%%%%%%%%%%%%%%%%%%%%%%%%%%%%%%%%%%%%%%%%

%
%%%%%%%%%%%%%%%%%%%%%%%%%%%%%%%%%%%%%%%%%%%%%%%%%%%%%%%%%%%%%%%%%%%%%%%%%%%%%%%%%%%%%%%%%%%%%%%%%%%%%%%%%%%
%%%%%%%%%%%%%%%%%%%%%%%%%%%%%%%%%%%%%%%%%%%%%%%%%%%%%%%%%%%%%%%%%%%%%%%%%%%%%%%%%%%%%%%%%%%%%%%%%%%%%%%%%%%

%\begin{center}
%	\includegraphics[width=150mm]{}
%\end{center}
%\caption{}
%\label{fig:one}
%\end{figure}

\par
\
\par\noindent
\noindent
{\bf Acknowledgments.}
YuH was supported in part by the Grant-in-Aid for Scientific Research (C) 20540133 and (B) 24340031 from 
Japan Society for the Promotion of Science. 
NK and IS also acknowledge financial supports of the Grant-in-Aid for Scientific Research (C) from
Japan Society for the Promotion of Science (Grant No. 24540116 and No. 23540176, respectively). 
%Acknowledgments.
\par
\
\par

\begin{small}
\bibliographystyle{jplain}

\end{small}

%\noindent\\
%\noindent\\
%\noindent\\

%\noindent {\large{\bf Appendix A}}  \\
%\noindent\\ 
%\noindent

% Appendix here

\end{document}